\documentclass[10pt]{article}
\usepackage{authblk}
\usepackage[utf8]{inputenc}   
\usepackage[T1]{fontenc}
\usepackage{lmodern}          
\usepackage[a4paper,margin=1in]{geometry}
\usepackage{microtype}
\usepackage{setspace}         
\usepackage{titling}

\usepackage{amsmath,amssymb,amsthm,mathtools}
\usepackage{bm}
\usepackage{physics}          

\usepackage{graphicx}
\graphicspath{{figs/}}        
\usepackage{caption}
\usepackage{subcaption}
\usepackage{booktabs}
\usepackage{siunitx}
\usepackage{algorithm}
\usepackage[noend]{algpseudocode}

\usepackage[numbers,sort&compress]{natbib} 
\usepackage{hyperref}
\hypersetup{
  colorlinks=true,
  linkcolor=blue,
  citecolor=blue,
  urlcolor=blue,
}
\usepackage[nameinlink,capitalize,noabbrev]{cleveref} 

\newtheorem{lemma}{Lemma}

\theoremstyle{definition}

\theoremstyle{remark}

\title{\textbf{Syntactic Structure, Quantum Weights}}
\author{Kentaro IMAFUKU}
\affil{National Institute of Advanced Industrial Science and Technology (AIST)}
\date{\empty}

\begin{document}
\maketitle

\begin{abstract}
Why do local actions and exponential Euclidean weights arise so universally in classical, statistical, and quantum theories?
We offer a structural explanation from minimal constraints on finite descriptions of admissible histories.
Assume that histories admit finite, self-delimiting (prefix-free) generative codes that can be decoded sequentially in a single forward pass.
These purely syntactic requirements define a minimal descriptive cost, interpretable as a smoothed minimal program length, that is additive over local segments.
First, any continuous local additive cost whose stationary sector coincides with the empirically identified classical variational sector is forced into a unique Euler--Lagrange equivalence class.
Hence the universal form of an action is fixed by descriptional structure alone, while the specific microscopic Lagrangian and couplings remain system-dependent semantic input.
Second, independently of microscopic stochasticity, finite prefix-free languages exhibit exponential redundancy: many distinct programs encode the same coarse history, and this redundancy induces a universal exponential multiplicity weight on histories.
Requiring this weight to be real and bounded below selects a real Euclidean representative for stable local bosonic systems, yielding the standard Euclidean path-integral form.
When Osterwalder--Schrader reflection positivity holds, the Euclidean measure reconstructs a unitary Lorentzian amplitude.
\end{abstract}


\section{Introduction}
\label{sec:intro}

How does a rule for describing an object influence the way the object appears?
Why do local actions and their exponential weights appear so universally,
viewed through the \emph{syntactic structure} of descriptions?

A striking commonality across classical mechanics, statistical field theory, and quantum theory is the appearance of a \emph{local additive
functional} on histories together with an \emph{exponential Euclidean weight}.  In Euclidean signature one typically encounters a functional
\begin{equation}
    S_E[x] \;=\; \int_{\tau_i}^{\tau_f}
      L_E\!\bigl(x(\tau),\dot x(\tau),\tau\bigr)\, d\tau ,
\end{equation}
and a weighting of the form
\begin{equation}
    P[x] \;\propto\; \exp[-S_E[x]/\hbar] .
\label{eq:intro-weight}
\end{equation}
These structures recur in systems with widely different microscopic content, suggesting that they may reflect general constraints on the
\emph{syntactic structure} by which histories are represented rather than detailed assumptions about particular dynamics.

This paper offers a structural account of the \emph{action--weight pair}.  We do not aim to change the empirical content of classical or
quantum theory, nor to propose new microscopic laws.  Instead we ask a more primitive question: \emph{what minimal, domain-independent
conditions on the describability of trajectories compel a local, additive functional together with its distinguished exponential
weight?}  Our results therefore concern the \emph{form} and universality of the action--weight structure, not the derivation of a unique
microscopic action or a numerical prediction of $\hbar$.

Analytic continuation from Lorentzian amplitudes $\exp(iS/\hbar)$ can motivate~\eqref{eq:intro-weight}\cite{osterwalder1973,osterwalder1975}, but such
arguments already presuppose a local action.  From a broader dynamical-systems viewpoint\cite{lanczos1949,arnold1989}, locality,
additivity, and variational structure are not inevitable: nonlocal update rules, global constraints, or dynamics lacking any action
principle are logically possible. Why, then, are admissible histories so widely organized by
(i) a local additive action and (ii) its exponential Euclidean weight?

\medskip
\noindent\textbf{Related approaches.}
Several programs address aspects of the action--weight pair. Stochastic approaches fall into two styles.
Nelson-type stochastic mechanics derives Schr\"{o}dinger dynamics from an assumed Brownian underpinning \cite{nelson1966,nelson1985},
while Parisi--Wu stochastic quantization yields Euclidean Gibbs weights as equilibrium measures of a fictitious Langevin time \cite{parisi1981}.
Both typically assume the action whose form they recover. Inference- or information-based approaches---maximum entropy \cite{jaynes1957a,jaynes1957b},
Fisher-information extremization \cite{frieden1998}, information-invariance \cite{reginatto1998,reginatto1999}, and
entropic dynamics \cite{caticha2011,caticha2012}---derive pieces of classical or quantum structure from constraints on inference, but usually take the admissible
class of continuum functionals as given---in particular, they start from a variational functional family (local, additive, continuum actions)
and then select or justify a member of that family by inference principles. Representation-first viewpoints
(e.g.\ Wheeler's ``It from Bit'' \cite{wheeler1989}, Tegmark's mathematical-universe hypothesis \cite{tegmark2008},
computational mechanics \cite{crutchfield1989,crutchfield2012}) have argued that physical regularities might reflect constraints on
representation or compression. 

We share this broad motivation, but our results hinge on a different, strictly syntactic starting point:
we impose minimal \emph{syntactic} well-formedness requirements on finite descriptions (prefix-freeness and sequential decodability), and show that
these alone \emph{force} two universal structures.
First, any continuous, additive, local descriptive cost compatible with minimal describability and with the empirically identified classical
stationary sector must lie in the Euler--Lagrange equivalence class of a local action.
Second, exponential redundancy of finite prefix-free languages induces a universal Euclidean exponential weight, without assuming microscopic
stochasticity.

\medskip
\noindent\textbf{Our perspective: histories are specified by finite, self-delimiting programs.}
We take a representation-theoretic premise \emph{prior} to any dynamical input:
admissible histories are treated as finitely describable objects, encoded by finite programs that are self-contained and stream-decodable, so that multiple
programs can be concatenated and parsed without external boundary markers. This minimal well-formedness requirement makes admissible descriptions
self-delimiting (in particular prefix-free) and online-decodable, with segment boundaries fixed by past symbols alone.  Accordingly,
\begin{quote}
\emph{Histories admit finite generative descriptions in a prefix-free,
sequentially decodable language.}
\end{quote}

This is a constraint on the \emph{description scheme}, not on microscopic physics: it fixes only the universal \emph{form} of admissible descriptive
costs, while leaving system-dependent content to later semantic/physical input. To model finite describability we work at fixed descriptive resolution and use
coarse-grained trajectories as analytic representatives; these are technical devices, not an effective-field-theory assumption.

Four domain-independent consequences follow. The first three are purely syntactic consequences of well-formed finite
descriptions; the fourth is an empirical anchoring that identifies the classical/stationary sector within that syntactic class:
\begin{itemize}
\item \textbf{local decomposition}: trajectories factor into typed local segments;
\item \textbf{local composition}: those segments recombine without ambiguity;
\item \textbf{sequential decodability}: segment boundaries are fixed by past information, with no lookahead;
\item \textbf{variational/classical correspondence}: coarse histories converge to differentiable trajectories.
\end{itemize}

\medskip
\noindent\textbf{From program length to a descriptive cost.}
Prefix-free syntax equips each trajectory with an integer minimal program length $|p|$. 
\emph{Prefix-free syntax constrains the \textbf{allowed form} of such descriptive costs, but does not by itself fix a unique one: the specific
cost $\ell$ realized in a given system reflects additional, system-dependent semantics.}
To admit variational analysis we smooth this into a continuous, segmentwise local functional $\ell[x]$, proportional (up
to a resolution-dependent scale) to minimal program length. By construction, $\ell[x]$ is a nonnegative descriptive \emph{cost}.

The key question is:
\begin{quote}
\emph{Given only syntactic locality and classical correspondence, which continuous additive local functionals can serve as descriptive costs?}
\end{quote}
The goal is not to extract a unique microscopic action from syntax, but to identify the universal \emph{form} shared by any finitely describable
dynamics with a well-defined classical/stationary sector.

Our first main result shows that these requirements force
\[
\ell[x]=\alpha\,S[x]+\text{(boundary terms)},\qquad \alpha\neq 0,
\]
where $S$ is a local additive functional in the Euler--Lagrange class selected by the stationary classical sector.
Here $\alpha$ is fixed only up to sign at the level of EL equivalence; the sign is fixed later when the redundancy weight is realized by a
bounded-below real representative. At this stage only the EL \emph{equivalence class} is fixed; no choice
of Euclidean or Lorentzian representative is assumed.

\medskip
\noindent\textbf{Redundancy of descriptions and the exponential weight.}
Finite generative languages exhibit exponential redundancy: for any finite prefix-free grammar with nonzero branching, the number
of syntactically valid strings of length $L$ grows as $\gamma^L$ for $\gamma>1$.  Many distinct programs thus encode the same coarse history.
Counting these redundant encodings yields a universal multiplicity weight
\[
    W[x]\propto\gamma^{-\ell[x]}
    =\exp[-(\ln\gamma)\,\ell[x]] .
\]
For this redundancy weight to be well formed on the same set of admissible histories where $\ell[x]$ is defined and nonnegative, the EL class must admit
a real bounded-below representative. For stable local bosonic systems, such a realization exists precisely in the real Euclidean sector; the redundancy construction therefore closes
on a Euclidean action $S_E$ and fixes the remaining sign freedom in $\ell=\alpha S+B$ to $\alpha>0$.
One then obtains
\[
    P[x]\propto \exp[-S_E[x]/\hbar_{\mathrm{eff}}],
    \qquad \hbar_{\mathrm{eff}}^{-1}=\alpha\ln\gamma .
\]
The exponential \emph{Euclidean} weight is therefore a consequence of redundant finite descriptions, not an independent physical postulate.
A direct identification of $\ell$ with the standard real-time Minkowskian action on Lorentzian histories is in general incompatible
with boundedness and plays no role in the construction.

\medskip
\noindent\textbf{Euclidean domain, OS positivity, and Lorentzian scope.}
The structural argument closes within the real Euclidean sector once redundancy is promoted to a positive, normalizable measure.
Osterwalder--Schrader (OS) reflection positivity \cite{osterwalder1973,osterwalder1975} is invoked only when one
wishes to reconstruct a unitary Lorentzian quantum theory from the resulting Euclidean measure; when OS positivity fails, the present
results still fix the Euclidean action--weight \emph{form} without guaranteeing its Lorentzian realization.

Most bosonic systems of interest admit real Euclidean representatives (e.g.\ scalar and gauge theories, as well as certain symmetry-reduced
(minisuperspace) gravitational models), so that their Euclidean path weights can be taken positive and reflection positive in the standard
OS setting \cite{glimm1987,rothe2012,jaffe2006}.
By contrast, fermionic systems and theories with intrinsically imaginary topological terms (such as generic $\theta$-terms or WZW-type
terms) typically lead to complex Euclidean effective actions or non-positive measures, and therefore require separate, case-specific
treatment before any probabilistic interpretation is available \cite{montvay1994,witten1979Theta,witten1983WZW,alexandru2022}.
While our framework flags these theories as requiring such scrutiny, we do not pursue a systematic classification here and leave it
to future work.

The overall logical flow of the paper is summarized in Fig.~\ref{fig:overview}. Starting from prefix-free, sequentially decodable generative programs at
finite descriptive resolution, we define the minimal program length $|p_x|:=\min_{p\to x}|p|$ and smooth it to a continuous segmentwise local
cost $\ell[x]$ at fixed descriptive resolution.
Section~\ref{sec:local} shows that any continuous additive local cost compatible with these syntactic requirements and with the classical
stationary sector must lie in the Euler--Lagrange equivalence class of a local action, $\ell[x]=\alpha S[x]+B$.
Section~\ref{sec:redundancy} then counts redundant programs encoding the same coarse trajectory;
exponential redundancy with rate $\Lambda=\ln \gamma$ implies a Euclidean weight
$P[x]\propto \exp(-S_E[x]/\hbar_{\mathrm{eff}})$ with $\hbar_{\mathrm{eff}}^{-1}=\alpha\Lambda$.
For theories satisfying OS reflection positivity, this Euclidean measure admits OS reconstruction to a unitary Lorentzian/Minkowskian amplitude
$\exp(+iS_M[x]/\hbar_{\mathrm{eff}})$. Section~\ref{sec:cosmo} finally provides a global consistency check by
comparing the emergent $\hbar_{\mathrm{eff}}$ with the cosmological value of $\hbar$.

\begin{figure*}[t]
  \centering
  \includegraphics[width=1\linewidth]{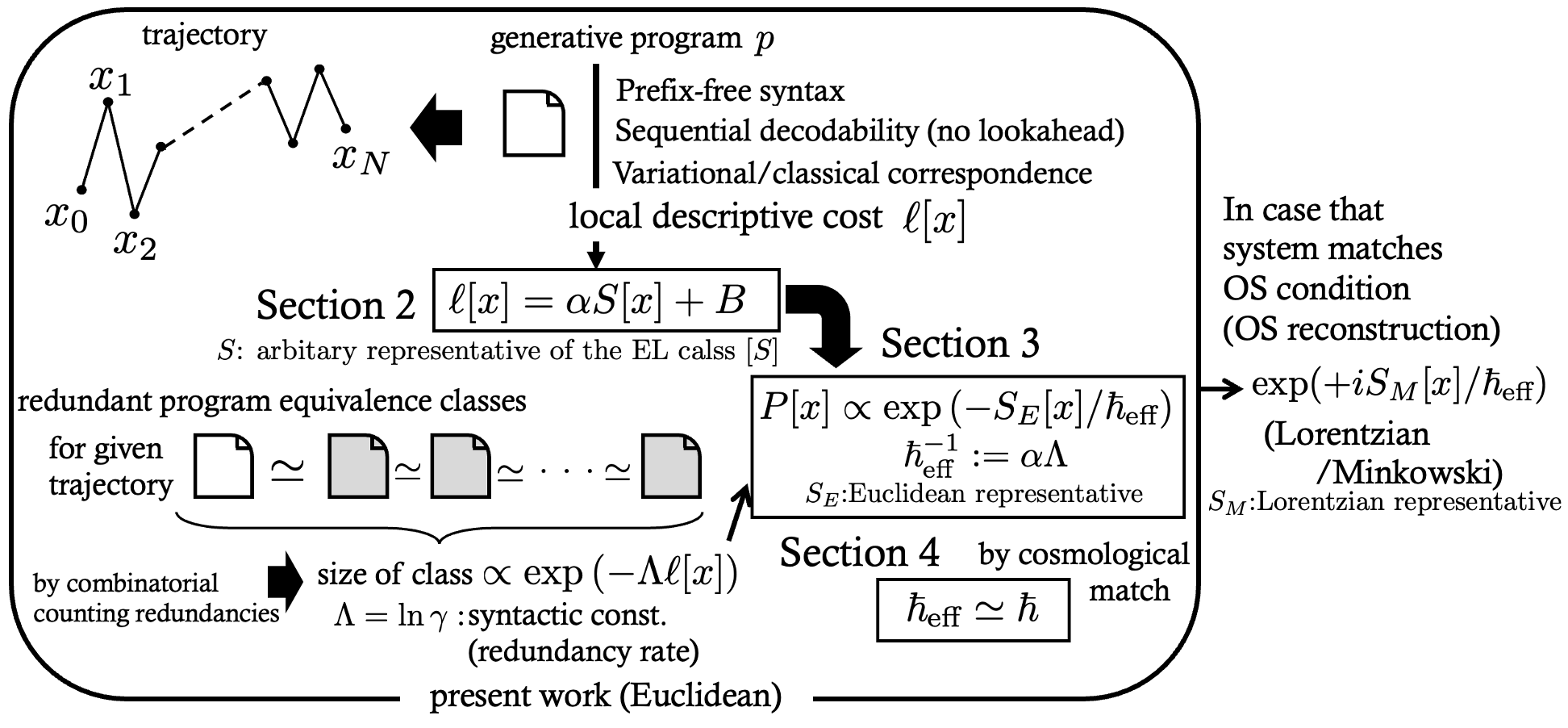}
\caption{Schematic overview of the syntactic framework and the structure of the paper.
Finite-resolution, prefix-free generative programs assign each trajectory a minimal program length $|p_x|:=\min_{p\to x}|p|$, and a continuous local
descriptive cost $\ell[x]$ obtained by smoothing $|p_x|$ at fixed descriptive resolution.
Section~\ref{sec:local} shows that syntactic locality together with variational/classical correspondence forces any continuous additive local
cost into the Euler--Lagrange equivalence class of a local action, $\ell[x]=\alpha S[x]+B$.
Section~\ref{sec:redundancy} counts redundant programs encoding the same coarse trajectory;
exponential redundancy with rate $\Lambda$ yields a Euclidean weight
$P[x]\propto \exp(-S_E[x]/\hbar_{\rm eff})$ with $\hbar_{\rm eff}^{-1}=\alpha\Lambda$.
When OS reflection positivity holds, OS reconstruction promotes this Euclidean theory to a unitary Lorentzian/Minkowskian amplitude
$\exp(+iS_M[x]/\hbar_{\rm eff})$. Section~\ref{sec:cosmo} compares the emergent $\hbar_{\rm eff}$ with the cosmological
value of $\hbar$, providing a global consistency check. The present work closes within the real Euclidean sector; Lorentzian
physics enters only conditionally through OS reconstruction.
\label{fig:overview}}
\end{figure*}

\medskip
\noindent\textbf{Cosmological calibration.}
The product $\alpha\Lambda$ remains undetermined by local structural arguments alone.  Laboratory systems offer no canonical global capacity,
but the observable Universe provides a unique reference scale through its holographic information content.  Matching descriptive cost to this
capacity yields $\hbar_{\mathrm{eff}}\simeq\hbar$ as a global consistency check rather than a microscopic derivation.

\medskip
\noindent\textbf{Summary.}
We propose that the robust pair
\[
(\text{local Euclidean action},\;\exp[-S_E/\hbar])
\]
arises from two generic structural features: (i) minimal syntactic constraints on admissible generative descriptions,
including sequential decodability; and (ii) exponential redundancy in finite prefix-free languages.
This perspective leaves standard physics unchanged while clarifying the universality of its action--weight structure and, for OS-positive
theories, its Lorentzian realization.

\section{Local descriptive costs and EL equivalence}
\label{sec:local}

We now make precise the representational premise announced in the Introduction.
Independently of any microscopic dynamics, we treat admissible histories as finitely describable objects specified by finite information.
To be finitely \emph{composable} and \emph{parsable} without external boundary markers, admissible programs must be self-contained and
stream-decodable in a single forward pass. Accordingly, we restrict attention to a prefix-free, sequentially
decodable generative language and to the minimal program length it induces on histories.

Prefix-freeness (self-delimitation) guarantees that a stream of programs can be uniquely segmented and concatenated without side information.
Sequential decodability requires that the boundary of each local segment be determined by past symbols alone, with no lookahead.
For brevity we refer to this no-lookahead property as \emph{syntactic causality}:
it is a statement about the \emph{parsing order of programs}, not a prior physical assumption about spacetime causality.

Given such a language at finite descriptive resolution, each history $x$ admits a minimal program length $|p_x|$.
Smoothing this integer length yields a continuous, additive local functional $\ell[x]$, which we call the descriptive cost.
The structural question is then:

\begin{quote}
\emph{Which continuous, additive, local functionals can serve as descriptive costs compatible with syntactic locality and with an
empirical classical stationary sector?}
\end{quote}

In this section we show that, under minimal assumptions motivated solely by syntactic locality and syntactic causality, any admissible descriptive
cost is forced to lie in a single Euler--Lagrange (EL) equivalence class. Thus an action-like functional appears not as a dynamical postulate but
as the unique continuous local \emph{EL form} compatible with prefix-free descriptive structure and the empirical classical limit.
No Euclidean or Lorentzian signature is selected here; only the EL class determined by the stationary sector is fixed.

The conceptual flow is:

\begin{itemize}
\item finite-resolution programs induce finite-resolution histories;
\item prefix-free syntax and syntactic causality induce a canonical segmentation into typed local pieces;
\item any admissible cost must therefore be local and additive;
\item empirical classical anchoring identifies a stationary sector of $\ell_n$ at each resolution;
\item a discrete EL-locality lemma then fixes the EL class of $\ell_n$, up to an overall factor and boundary terms.
\end{itemize}

Throughout this section continuum notation is used only as a stable analytic representative of fixed finite-resolution descriptions.

\medskip
\noindent
\textit{Remark (scope).}
For notational simplicity we present the argument for non-relativistic particle trajectories.
The proof relies only on (i) a local segmentation induced by prefix-free, sequentially decodable
syntax and (ii) variational/classical correspondence at fixed descriptive
resolution. Once a covariant notion of local cells/segments is specified, the same EL-equivalence and
locality conclusions extend verbatim to relativistic bosonic field theories.
A relativistic formulation is given in Appendix~\ref{app:relativistic}.

\subsection{Finite-resolution descriptions}

The resolution index $n$ models finite describability. It is not an effective-theory cutoff and makes no assumption about
intermediate EFT structure.

A program in a finite alphabet encodes only finite numerical information per generative step.
We model this by introducing a resolution index $n$ that fixes:
\begin{itemize}
\item a temporal discretization scale, and
\item the numerical precision at which each node is stored.
\end{itemize}

A resolution--$n$ history is a finite sequence
\[
x^{(n)}=(x^{(n)}_0,\ldots,x^{(n)}_{N_n}),
\qquad x^{(n)}_k\in\mathbb{R}^d,
\]
where each coordinate is stored to $n$ digits of precision in the program but treated analytically as a real variable.
Adjacent pairs define segments
\[
s_k^{(n)}:=(x^{(n)}_{k-1},x^{(n)}_k).
\]

This cleanly separates the finite syntactic layer from the smooth variational layer while retaining compatibility with discrete
variational calculus.

\subsection{Prefix-free and boundary-delimited programs}

A program generating $x^{(n)}$ has the form
\[
p=h(s_1^{(n)})\,q_1\,h(s_2^{(n)})\,q_2\cdots
  h(s_{N_n}^{(n)})\,q_{N_n},
\]
where each header $h(s_k)$ is drawn from a finite prefix-free set $\mathcal{H}_n$ and declares the \emph{syntactic type} of the segment,
while each payload $q_k$ carries its numerical data.

Prefix-freeness ensures instantaneous decodability: segment boundaries are recognizable from past symbols alone.
Appendix~\ref{app:prefixfree} shows that prefix-freeness and boundary delimitation are in fact necessary for
any well-defined local additive cost; without them segmentation is ambiguous.

The raw program length is
\[
|p|=\sum_{k=1}^{N_n}\bigl(|h(s_k)|+|q_k|\bigr).
\]

\subsection{Generalized descriptive cost}

For variational analysis we introduce a smoothed segmentwise local cost functional
\[
\ell_n[p]=\sum_{k=1}^{N_n}\ell^{\mathrm{loc}}_n(s_k^{(n)}),
\]
subject to:
\begin{itemize}
\item[(L1)] $\ell_n^{\mathrm{loc}}(s)$ is smooth and positive;
\item[(L2)] $\ell_n[p]$ is linearly equivalent to raw length: $A_n^{-1}|p|\le\ell_n[p]\le A_n|p|$;
\item[(L3)] $\ell_n^{\mathrm{loc}}$ depends only on the segment and its syntactic type.
\end{itemize}

The path cost is the minimum over programs generating the history:
\[
\ell_n[x^{(n)}]:=\min_{p\to x^{(n)}}\ell_n[p].
\]
By construction $\ell_n[x^{(n)}]\ge0$ is a descriptive \emph{cost}.

\subsection{Structural and anchoring assumptions}

Prefix-free locality implies that any admissible representation decomposes into local segments and recombines unambiguously.
We formalize this with two structural assumptions and two anchoring assumptions.

\paragraph{A1 (Additivity).}
\[
\ell_n[x^{(n)}]
=\sum_{k=1}^{N_n}\ell^{\mathrm{loc}}_n(s_k^{(n)}).
\]

\paragraph{A2 (Local stability).}
$\ell^{\mathrm{loc}}_n$ is smooth in its arguments and bounded below on compact sets, ensuring well-defined discrete EL operators.

\paragraph{A3 (Variational/Classical correspondence: empirical anchoring).}
Let $\mathcal{C}^{(n)}_{\rm emp}$ be the set of empirically identified coarse histories at resolution $n$.
We assume
\[
\mathcal{C}^{(n)}_{\rm emp}=\mathrm{Stat}(\ell_n),
\]
i.e.\ empirical coarse histories coincide (as a set in the abstract history space) with the stationary set of $\ell_n$.

\paragraph{A4 (Syntactic universality).}
Every admissible segment has at least one valid encoding in the grammar.

\medskip
\noindent
\textit{Remark.}
A3--A4 provide the minimal anchoring to the empirical classical sector. 
Appendix~\ref{app:strong-weight-classical} shows that A3 may equivalently be viewed as identifying empirical coarse histories with the dominant
stationary sector in a strong-weight/large-cost regime of the syntactic multiplicity measure.

\subsection{From the stationary sector to an EL class}

A central result in discrete analytical mechanics states that any second-order local difference equation admits a variational
formulation. Thus A3 guarantees the existence of a discrete local additive functional
$\mathcal{S}_n$ with local density $L^{\mathrm{loc}}_n$ such that
\[
\mathrm{Stat}(\mathcal{S}_n)=\mathrm{Stat}(\ell_n).
\]
Importantly, $\mathcal{S}_n$ is introduced only \emph{after} the stationary set is fixed; locality and variational structure are not
assumed at the outset. At this stage $\mathcal{S}_n$ is merely a representative of the EL class determined by
the stationary sector, with no signature specified.

\subsection{EL-locality lemma}

We next use a discrete EL-locality lemma (proved in Appendix~\ref{app:EL-locality-proof}):

\begin{lemma}[EL-locality]
Under A1 and A2, if $\ell_n$ and $\mathcal{S}_n$ share the same stationary histories,
then there exist $\alpha_n\neq0$ and a function $G_n$ such that
\[
\ell_n^{\mathrm{loc}}(s)
=\alpha_n\,L_n^{\mathrm{loc}}(s)+G_n(x_k)-G_n(x_{k-1})
\]
for every segment $s=(x_{k-1},x_k)$.
\end{lemma}

Summing over segments gives
\[
\ell_n[x^{(n)}]
=\alpha_n\,\mathcal{S}_n[x^{(n)}]+B_n(x_f,x_i),
\qquad
B_n(x_f,x_i)=G_n(x_f)-G_n(x_i).
\]
Hence $\ell_n$ and $\mathcal{S}_n$ lie in the same local EL equivalence class, written $[\mathcal{S}_n]$.

\medskip
\noindent
\textit{Orientation of the EL class.}
Because $\ell_n$ is a smoothed minimal program length, it is nonnegative as a descriptive cost.
This fixes the orientation of the EL class: we choose the representative $\mathcal{S}_n$ within $[\mathcal{S}_n]$ so that
$\ell_n=\alpha_n\mathcal{S}_n+B_n$ with $\alpha_n>0$. This step uses only the representational meaning
of $\ell_n$ as a cost. No boundedness, Euclidean-sector, or real-measure requirement is imposed on $\mathcal{S}_n$
at this stage; these conditions appear only in Section~\ref{sec:redundancy}.

\subsection{Continuum limit}

As $n\to\infty$, admissible coordinate values become dense and the families $\{\ell_n,\mathcal{S}_n\}$ converge to continuum representatives
$\ell[x]$ and $\mathcal{S}[x]$:
\[
\ell[x]=\alpha\,\mathcal{S}[x]+B(x_f,x_i),
\qquad \alpha\neq0.
\]

\emph{Thus Section~\ref{sec:local} fixes only the EL class $[\mathcal{S}]$ determined by the empirical stationary sector,
leaving open which real bounded-below representative (if any) will later be selected when redundancy induces a weight.}

\medskip
\noindent\textbf{Notational convention.}
Below we suppress the label $n$ and write $\ell[x]$ and $\mathcal{S}[x]$ for simplicity.
Unless stated otherwise, the structural results above do not require the strict limit $n\to\infty$: at any sufficiently fine fixed resolution the
descriptive cost remains in the same EL class up to the same overall factor and boundary terms.


\section{From Syntactic Redundancy to Quantum Weights}
\label{sec:redundancy}

Section~\ref{sec:local} fixed the \emph{Euler--Lagrange (EL) equivalence class} of the descriptive cost by matching stationary histories.
Concretely, every admissible path $x$ has a finite prefix-free, boundary--delimited description reflecting syntactic locality and
syntactic causality, and its minimal descriptive cost $\ell[x]$ lies in a unique local EL class:
\begin{equation}
\ell[x] \;=\; \alpha\,S[x] \;+\; B(x_f,x_i),
\label{eq:ell-S-proportional}
\end{equation}
where $\alpha\neq 0$ is a grammar--dependent proportionality constant, $B(x_f,x_i)$ is a boundary term depending only on endpoints, and
$S$ denotes an \emph{arbitrary local representative} in the EL class determined by the empirical stationary sector.
At the EL level no signature or boundedness is assumed: Section~\ref{sec:local} fixes only the class, not a distinguished member.

Here we take a logically new step.  We first fix the coarse-grained history space (i.e., the equivalence relation that defines when two descriptions represent the same admissible history).  ``Junk'' then refers to the redundant prefix-free programs within each such equivalence class.
Finite prefix-free languages contain a redundant syntactic sector: many distinct programs can encode the same coarse history.
We show that \emph{multiplicity counting alone} over this redundant sector induces a universal exponential redundancy weight
\begin{equation}
P[x] \;\propto\; \exp[-\Lambda\,\ell[x]],
\label{eq:P-exp-l-general}
\end{equation}
with redundancy exponent $\Lambda>0$ determined purely by the grammar.

\noindent\emph{Comment (what is new in this section).}
The exponential factor \eqref{eq:P-exp-l-general} is a purely structural multiplicity factor, derived from redundancy counting together with the
nonnegativity of $\ell$ as a descriptive cost. If one further requires this factor to define a \emph{real, finite measure over admissible histories},
then the EL class fixed in Section~\ref{sec:local} must admit a real bounded-below representative.
This realizability requirement does \emph{not} modify the EL class; it only selects a member of it.  For stable bosonic systems, that member
is uniquely Euclidean. Euclidean signature is thus \emph{selected}, not assumed.

The mechanism parallels standard counting arguments in information theory. In a finite prefix-free language the number of admissible strings of length $L$ grows exponentially \cite{cover2006,mcmillan1953}, and algorithmic thermodynamics derives Gibbs-like weights by counting programs under a length ceiling  \cite{imafuku_at}. Here the same exponential redundancy acts as a syntactic reservoir: summing over all compatible redundant (``junk'') completions induces a universal multiplicity weight on coarse histories, without assuming microscopic stochasticity.
\subsection{Redundant strings and an arbitrary length ceiling}

Physical (``core'') programs and redundant (``junk'') programs coexist within the same prefix-free syntax, reflecting a shared representational
infrastructure. Any full description naturally decomposes as
\[
p_{\mathrm{tot}} = p_{\mathrm{core}}\,p_{\mathrm{junk}},
\qquad \ell[p_{\mathrm{tot}}] = \ell[x] + \ell_{\mathrm{junk}}.
\]

To quantify the multiplicity of junk strings compatible with~$x$, we introduce an \emph{arbitrary length ceiling}~$K$:
\begin{equation}
\ell[x] + \ell_{\mathrm{junk}} \;\le\; K.
\label{eq:formal-K}
\end{equation}
The parameter~$K$ serves only to define the ensemble; all relative weights will be independent of~$K$, and the limit $K\to\infty$
does not affect the result.

Because junk strings inhabit a finite prefix-free language (and hence a finite-branching code tree), the number of admissible strings of
length~$L$ grows exponentially:
\begin{equation}
N_{\mathrm{junk}}(L) \;\sim\; \gamma^{L},\qquad \gamma>1.
\label{eq:junk-growth}
\end{equation}
This growth rate is a purely syntactic property of the junk sector and does not involve assumptions about the physical encoding.

\subsection{Induced measure from redundancy}

For a fixed physical path~$x$,
\[
L_{\max}(x)=K-\ell[x]
\]
is the longest permissible junk length. The number of compatible junk programs is therefore
\begin{equation}
W[x]\;\equiv\;\sum_{L\le L_{\max}(x)} N_{\mathrm{junk}}(L)
\;\approx\; \frac{\gamma^{L_{\max}(x)+1}}{\gamma-1},
\label{eq:W-def}
\end{equation}
where subleading $x$--independent terms are omitted.

Up to an $x$--independent factor,
\[
W[x]\;\propto\;\gamma^{\,L_{\max}(x)}\;=\;\gamma^{K-\ell[x]}\;=\;\gamma^{K}\,\gamma^{-\ell[x]}.
\]
Since $\gamma^K$ cancels under relative weighting, the redundancy factor is
\begin{equation}
P[x]\;\propto\;\gamma^{-\ell[x]} \;=\;\exp\!\left[-(\ln\gamma)\,\ell[x]\right].
\label{eq:P-exp-l}
\end{equation}

Only the redundancy rate of the junk sector enters here; the physical (core) sector may have a distinct alphabet or code rate.
In this sense the junk sector acts as a \emph{redundant syntactic reservoir}: marginalizing over it induces a Gibbs-like weight over coarse histories.

For convenience we introduce the redundancy exponent
\[
\Lambda \equiv \ln\gamma \;>\;0.
\]

\subsection{Bounded-below realization of the EL class}

Equation~\eqref{eq:P-exp-l} is purely structural: it follows from redundancy in a finite prefix-free language and holds for any descriptive
cost $\ell$ in the EL class fixed in Section~\ref{sec:local}. Requiring \eqref{eq:P-exp-l} to define a \emph{real finite measure} over admissible
histories \emph{then} forces the EL class to admit a representative compatible with $\ell[x]\ge 0$.
Since $\Lambda>0$, this requires a representative $S$ that is real and bounded below.

For stable local bosonic systems, a real bounded-below realization exists precisely in the Euclidean sector.
The redundancy construction therefore selects a real Euclidean representative $S_E$.
(When no such real bounded-below Euclidean representative exists---as in fermionic systems or theories with intrinsically imaginary topological
terms---the redundancy factor \eqref{eq:P-exp-l} remains a well-defined structural object, but its statistical interpretation requires additional,
case-specific input.)

With this Euclidean realization,
\begin{equation}
\ell[x] \;=\; \alpha\,S_E[x] \;+\; B(x_f,x_i),
\label{eq:ell-SE-proportional}
\end{equation}
and the remaining sign ambiguity in~\eqref{eq:ell-S-proportional} is fixed to $\alpha>0$.

\subsection{Emergence of the Euclidean quantum weight}

Substituting~\eqref{eq:ell-SE-proportional} into \eqref{eq:P-exp-l} yields
\[
P[x]\;\propto\;\exp\!\Bigl[-\Lambda\,(\alpha\,S_E[x] + B(x_f,x_i))\Bigr].
\]
The boundary term depends only on endpoints and therefore drops out of relative weights, giving
\begin{equation}
P[x]\;\propto\;\exp\!\left[-(\alpha\Lambda)\,S_E[x]\right].
\end{equation}

The emergent quantum scale thus satisfies
\[
\hbar_{\mathrm{eff}}^{-1} = \alpha\Lambda,
\]
yielding the Euclidean weight in standard form
\begin{equation}
P[x]\;\propto\;
\exp\!\left[-\frac{S_E[x]}{\hbar_{\mathrm{eff}}}\right].
\end{equation}

The Euclidean quantum measure therefore arises directly from \emph{syntactic redundancy} in a finite prefix-free generative language,
without invoking microscopic stochasticity or additional dynamical assumptions. Euclidean signature is not assumed; it is the unique real bounded-below
realization compatible with the redundancy weight and the cost meaning of~$\ell$.

\medskip
\noindent\textit{Remark (scale redundancy).}
The structural results of Sections~\ref{sec:local} and~\ref{sec:redundancy} depend on $\ell$, $\Lambda$, and $\alpha$ only through the product
$\alpha\Lambda$. Indeed, for any $c>0$ the simultaneous rescaling
\[
  \ell \;\to\; c\,\ell,\qquad
  \Lambda \;\to\; \Lambda/c,\qquad
  \alpha \;\to\; c\,\alpha
\]
leaves both the redundancy weight $P[x]\propto\exp[-\Lambda\,\ell[x]]$ and the emergent scale $\hbar_{\mathrm{eff}}^{-1}=\alpha\Lambda$ invariant.
Thus the overall normalization of the descriptive cost and of the redundancy exponent is a matter of convention; only their product is
structurally fixed. The cosmological discussion in Section~\ref{sec:cosmo} can therefore be viewed as choosing a particular normalization
within this scale freedom, using the Universe as an absolute reference system.

\subsection{Lorentzian scope}

Once a real Euclidean measure is obtained, Lorentzian quantum theory can enter only \emph{conditionally}.
For theories satisfying OS reflection positivity, the Euclidean measure reconstructed from~\eqref{eq:P-exp-l} admits OS
reconstruction to a unitary Lorentzian amplitude of the usual form $\exp(+iS_M/\hbar_{\mathrm{eff}})$.
When OS positivity fails, the present mechanism still explains the universality of the Euclidean action--weight \emph{form} but does not by
itself guarantee a Lorentzian realization.

\subsection{Summary}

\begin{itemize}
\item
Section~\ref{sec:local} fixed only the EL class of admissible descriptive costs by matching stationary histories; no signature or
boundedness was assumed there.

\item
Exponential redundancy in the junk sector of any finite prefix-free language implies that the multiplicity of descriptions
compatible with $x$ scales as $W[x]\propto\gamma^{-\ell[x]}$.

\item
Hence redundancy induces the universal exponential factor $P[x]\propto\exp[-\Lambda\,\ell[x]]$, independent of the arbitrary ceiling $K$.

\item
Requiring this factor to define a real normalizable weight selects a real bounded-below representative in the EL class.
For stable local bosonic systems this representative is uniquely Euclidean, giving
\[
P[x]\propto \exp[-S_E[x]/\hbar_{\mathrm{eff}}],\qquad\hbar_{\mathrm{eff}}^{-1}=\alpha\Lambda.
\]
\end{itemize}

Thus the Euclidean quantum weight emerges not as an additional postulate, but as a robust structural consequence of finite prefix-free generative
syntax and its intrinsic exponential redundancy.

For gauge systems, physical histories are gauge orbits rather than individual configurations. A finite prefix-free program, however, must output concrete configurations step by step, so any syntactic description implicitly fixes a representative on each orbit (a discrete analogue of gauge fixing). Counting the orbit–internal multiplicity then yields an additional orbit-dependent prefactor--analogous to a Faddeev-Popov determinant \cite{faddeev1967,faddeev1973}--without altering the universal exponential redundancy weight. Appendix \ref{app:FP-toy} illustrates this with a minimal toy model.


\section{Cosmological consistency check of the emergent Planck scale}
\label{sec:cosmo}

Sections~\ref{sec:local}--\ref{sec:redundancy} showed that a finite, prefix-free, sequentially decodable description language
induces a universal redundancy weight on admissible histories,
\[
P[x]\propto\exp\!\left[-\frac{S_E[x]}{\hbar_{\mathrm{eff}}}\right],
\qquad
\hbar_{\mathrm{eff}}^{-1}=\alpha\Lambda,
\]
where $\alpha>0$ relates the minimal descriptive cost to a Euclidean representative in the EL class, and $\Lambda=\ln\gamma$ is
the redundancy exponent of the grammar.  The syntactic framework fixes the \emph{form} of the weight
but leaves the absolute scale $\alpha\Lambda$ undetermined. This section provides an external calibration of that scale.

\medskip
\noindent\textbf{Status of this section.}
One could fix $\hbar_{\mathrm{eff}}$ by directly matching the emergent weight to the standard Euclidean path-integral form, i.e.\ by imposing
$\hbar_{\mathrm{eff}}=\hbar$ as empirical input.  While legitimate, such a direct matching would not constitute an independent check on the
product $\alpha\Lambda$.  Instead, we calibrate $\alpha\Lambda$ using only classical de~Sitter estimates and the holographic information bound,
without assuming the Euclidean quantum measure whose normalization is at issue.

\medskip
\noindent\textbf{Why cosmology can provide an absolute calibration.}
For ordinary subsystems there is no invariant notion of total descriptive capacity: the representational environment can always be
enlarged, so no canonical upper bound on total program length exists. 
Consequently $\alpha\Lambda$ cannot be fixed from laboratory systems alone.

The observable Universe is exceptional.
Its cosmic event/hubble horizon provides a geometric, essentially model-independent upper bound on total
information content via the holographic principle \cite{bekenstein1973,hawking1975,thooft1993,susskind1995}.
Thus cosmology is the only known setting where a \emph{finite absolute} descriptive budget can be compared to a \emph{finite absolute}
(on-shell) Euclidean action.
We therefore compare $\ell[\mathrm{univ}]$ with $S_{E,\mathrm{univ}}$ and treat the result as an order-of-magnitude
consistency check on $\alpha\Lambda$, not as a first-principles derivation of $\hbar$.

\medskip

\subsection{Physical assumptions}

We restrict attention to standard late-time cosmology and list the inputs used in the calibration. 
These are classical/holographic assumptions and add no dynamical hypotheses beyond those already fixed
by the syntactic analysis:

\begin{description}

\item[(U1) Flat, critical-density FRW cosmology.]
We adopt a spatially flat FRW universe with present Hubble parameter $H_0$ and total energy density
fixed at the critical value $\rho_c=3H_0^2/(8\pi G)$ \cite{mukhanov2005,dodelson2003}.

\item[(U2) Holographic information capacity.]
The observable Universe is bounded by a Hubble horizon of radius $R_H=H_0^{-1}$ and area
\[
A_H=4\pi R_H^2.
\]
Its maximal information capacity is the holographic value
\[
I_{\mathrm{holo}}
=\frac{A_H}{4\ell_p^2}
=\frac{\pi}{\ell_p^2 H_0^2},
\qquad
\ell_p^2=G\hbar.
\]

\item[(U3) On-shell Euclidean action of de~Sitter.]
Euclidean de~Sitter space has period $\beta_{\mathrm{dS}}=2\pi/H_0$ \cite{gibbons1977}.
The total energy within the Hubble volume is
\[
E_U
=\rho_c\frac{4\pi}{3H_0^3}
=\frac{1}{2GH_0}
=\frac{\hbar}{2\ell_p^2H_0}.
\]
Thus the on-shell Euclidean action associated with the de~Sitter horizon is
\[
S_{E,\mathrm{univ}}
  =\beta_{\mathrm{dS}}E_U
  =\hbar\,I_{\mathrm{holo}}.
\]
(Throughout this section $S_{E,\mathrm{univ}}$ denotes this standard de~Sitter/CEH on-shell Euclidean action.)

\item[(U4) Self-description as a calibration ansatz.]
Let $\ell[\mathrm{univ}]$ denote the minimal prefix-free description of
the Universe's \emph{entire admissible Euclidean history} at the given descriptive resolution.
This is a self-contained program generating a full history in the abstract admissible-history space;
it is \emph{not} a record of what any internal observer will eventually access.

Because the Euclidean Universe is compact in imaginary time, boundary terms in~\eqref{eq:ell-SE-proportional} cancel automatically.
At cosmological scales it is natural to assume that all sectors of the Universe---dynamical rules, initial data, and any redundant syntactic
structure---are encoded using a \emph{common} finite descriptive resource, hence share a single redundancy exponent $\Lambda$.
We therefore adopt the calibration ansatz that the Universe's minimal self-description generically occupies its holographic
capacity up to that coding-efficiency scale:
\[
\ell[\mathrm{univ}] = \frac{I_{\mathrm{holo}}}{\Lambda}.
\]
Allowing a constant usage fraction $\eta\le 1$ rescales this to
\[
\ell[\mathrm{univ}] = \frac{\eta\,I_{\mathrm{holo}}}{\Lambda}.
\]
Further motivation for viewing (U4) as a natural consistency condition for a closed self-describing Universe is given in
Appendix~\ref{app:self-description}.
\end{description}

Taken together, (U1)--(U4) permit an absolute comparison of the total syntactic cost of the Universe with its de~Sitter Euclidean action.

\subsection{Fixing the proportionality constant}

From Section~\ref{sec:local},
\[
\ell[\mathrm{univ}] = \alpha\,S_{E,\mathrm{univ}}.
\]

Combining (U3) and (U4) gives
\[
\alpha S_{E,\mathrm{univ}}
  = \frac{I_{\mathrm{holo}}}{\Lambda},
\qquad
S_{E,\mathrm{univ}} = \hbar\,I_{\mathrm{holo}}.
\]
Eliminating $I_{\mathrm{holo}}$ yields
\[
\alpha\hbar=\frac{1}{\Lambda},
\qquad\Rightarrow\qquad
\alpha=\frac{1}{\Lambda\,\hbar}.
\]
With a usage fraction $\eta<1$,
\[
\alpha=\frac{1}{\eta\,\Lambda\,\hbar}.
\]

\subsection{Determination of $\hbar_{\mathrm{eff}}$}

From Section~\ref{sec:redundancy},
\[
\hbar_{\mathrm{eff}}^{-1}=\alpha\Lambda.
\]
Substituting the cosmologically determined value of $\alpha$ gives
\[
\hbar_{\mathrm{eff}}^{-1}
  = \Lambda \cdot \frac{1}{\Lambda\hbar}
  = \frac{1}{\hbar},
\qquad\Rightarrow\qquad
\hbar_{\mathrm{eff}}=\hbar.
\]
With a usage fraction $\eta<1$, one finds
$\hbar_{\mathrm{eff}}=\hbar/\eta$.

\subsection{Interpretation}

The equality $\hbar_{\mathrm{eff}}=\hbar$ is not a derivation of Planck’s constant.
Rather, it shows that the syntactic--redundancy mechanism admits a natural absolute calibration: if a closed Universe uses its holographic
descriptive capacity up to its intrinsic coding efficiency,
then the resulting emergent quantum scale is consistent with the observed value of $\hbar$.

\subsection{Summary}

\begin{itemize}
\item The emergent quantum scale satisfies $\hbar_{\mathrm{eff}}^{-1}=\alpha\Lambda$.
\item The observable Universe uniquely provides both a total information capacity (via the holographic bound) and a horizon-normalized Euclidean on-shell action.
\item Adopting near-saturation of holographic capacity as a calibration ansatz fixes $\alpha=1/(\Lambda\hbar)$ (up to $\eta$).
\item Consequently $\hbar_{\mathrm{eff}}=\hbar$ (up to order-unity corrections from incomplete capacity usage).
\end{itemize}

In short, the Planck scale appears here not as a prediction, but as the unique value compatible with de~Sitter action, holographic capacity, and
the redundancy structure of finite prefix-free generative descriptions of the Universe’s admissible history.


\section{Discussion and Outlook}
\label{sec:discussion}

The analysis presented here supports a unified perspective in which classical dynamics, quantum weighting, and even cosmological scales
reflect structural properties of finite generative descriptions rather than independent dynamical postulates.
The central shift is representational: the action--weight pair need not be taken as a primitive element of physics.
Instead, minimal well-formedness requirements on finite, prefix-free, sequentially decodable descriptions fix a universal
\emph{structural class} of admissible local descriptive costs and their redundancy-induced weights.

A key point is a clean separation between two layers. At the \emph{universal syntactic layer}, representation determines the
Euler--Lagrange (EL) equivalence class of admissible local additive costs and the form of the redundancy weight.
At the \emph{system-specific semantic/physical layer}, different microscopic theories select particular representatives within that class
(specific Lagrangians, couplings, and sectors). Euclidean signature is not assumed in the EL-class determination; it
appears only when the redundancy weight is required to be a real, bounded-below measure compatible with $\ell$ as a nonnegative cost.
Lorentzian quantum theory is not assumed at any stage; it enters only conditionally, when OST reflection positivity allows the resulting real
Euclidean measure to reconstruct to a unitary Lorentzian amplitude.

In short: \emph{descriptive syntax fixes the EL class of costs; redundancy fixes the weight form; boundedness selects a real Euclidean
representative; system-specific physics chooses a member of that class; and cosmic information bounds calibrate the overall scale.}

This section summarizes the conceptual structure, clarifies interpretational points, and indicates several directions for further development.

\subsection{Summary of conceptual structure}

The results can be grouped into four mutually reinforcing components:

\begin{itemize}

\item \textbf{(Syntactic locality and causality).}
A finite prefix-free, boundary--delimited grammar enforces a canonical decomposition of any finite-resolution trajectory into locally typed
segments. Any compatible continuous cost functional must therefore decompose into segmentwise contributions.
This locality is representational---a constraint on how histories can be described---rather than a dynamical assumption.

\item \textbf{(Variational/Classical correspondence).}
Requiring that empirically identified coarse trajectories be stationary points of the descriptive cost forces that cost to lie in a unique local
EL equivalence class:
\[
\ell[x]=\alpha S[x]+B,\qquad \alpha>0.
\]
At this level only the EL class is fixed, not a specific signature or microscopic representative.
Thus an action-like local functional appears not as an axiom, but as the unique additive local form compatible with descriptive locality and the
classical stationary sector.

\item \textbf{(Statistical redundancy).}
Finite generative grammars exhibit exponential redundancy: the number of admissible strings of length $L$ grows as $\gamma^L$ with
$\gamma>1$. Counting redundant descriptions of a fixed physical path gives the structural multiplicity factor
\[
P(x)\propto\gamma^{-\ell[x]}
=\exp[-\Lambda\,\ell[x]],\qquad \Lambda=\ln\gamma>0.
\]
For this multiplicity factor to define a well-formed weight on the same history space as the cost $\ell[x]\ge 0$, the EL class must be identified
by a real representative bounded below. For stable local bosonic systems such a realization exists precisely in
the real Euclidean sector, so the redundancy construction closes on a Euclidean action $S_E$ and yields the Euclidean quantum weight
$\exp[-S_E[x]/\hbar_{\mathrm{eff}}]$ with $\hbar_{\mathrm{eff}}^{-1}=\alpha\Lambda$.

\item \textbf{(Cosmological calibration).}
Identifying the descriptive capacity required to encode the Universe with its holographic information bound fixes $\alpha\Lambda$ in standard
cosmology, leading to $\hbar_{\mathrm{eff}}\simeq\hbar$. This is a global consistency check of the syntactic mechanism, not a
derivation of Planck’s constant.

\end{itemize}

\medskip
\noindent\textit{Scope.}
The first three components close entirely within the real Euclidean sector and require no OS assumptions.
OS reflection positivity is invoked only when one wishes to promote the resulting Euclidean measure to a unitary
Lorentzian quantum theory; when OS positivity fails, the present results still fix the Euclidean action--weight \emph{form}, but Lorentzian
unitarity is not structurally ensured.

Taken together, these components indicate that the Euclidean weight is not an auxiliary postulate but a \emph{structural consequence} of
descriptive locality, variational/classical correspondence, and redundancy in finite prefix-free generative languages.

\subsection{Interpretational comments}

The framework admits a compact conceptual summary:

\begin{quote}
\emph{Classical stationary structure and Euclidean quantum weights arise from informational constraints
imposed by finite, prefix-free generative syntax.}
\end{quote}

Two clarifications are especially important.

\medskip
\noindent\textbf{EL class vs.\ signature selection.}
The syntactic argument determining the EL class of $\ell$ is \emph{independent of any choice of signature}: it refers only to local
additive costs and their stationary sector. Signature enters only when the redundancy weight is required to be
real and bounded below. Because $\ell$ is a nonnegative descriptive cost and the redundancy rate
$\Lambda$ is positive, boundedness and normalizability imply that only a real Euclidean representative can realize the EL class compatibly.
A direct identification of $\ell$ with the standard real-time Minkowskian action on Lorentzian histories is in general incompatible with this
boundedness and is neither required nor implied by the framework.

\medskip
\noindent\textbf{Weight reflects multiplicity, not ignorance.}
No epistemic uncertainty is introduced. A physical trajectory may be encoded by a single minimal program.
Quantum weights arise because many syntactically distinct programs generate the same coarse history.
The weight reflects multiplicity within the descriptive ensemble, not lack of knowledge about the underlying trajectory.

\medskip
\noindent\textbf{Continuum not fundamental.}
Resolution-$n$ discrete trajectories are the basic objects. The action-like functional appears only as the stable continuum
representative of syntactically constrained discrete costs. Thus quantum structure emerges from redundancy in a finite grammar, not
from a pre-imposed continuum ontology.

\medskip
\noindent\textbf{Universality of $\hbar$.}
The cosmological argument shows that the syntactic mechanism \emph{can} be calibrated against a unique global information bound.
It does not by itself explain why all subsystems share the same $\hbar$;
that universality is naturally interpreted as a consequence of subsystems inheriting their descriptive resources
from the global cosmic grammar (Appendix~\ref{app:self-description}).

\paragraph*{Relation to Solomonoff priors.}

Given $\ell[x]=\alpha S_E[x]+B$, one might consider the Solomonoff prior $P(x)\propto2^{-\ell[x]}$ \cite{solomonoff1964a, solomonoff1964b,LiVitanyi2008}.  
Although this leads formally to an exponential weight, it reverses the logical order.
Solomonoff’s prior \emph{assumes} a coding prior $2^{-|p|}$; the choice of base~2 is conventional.
In contrast, the present framework does not posit a prior: the exponential weight is a \emph{derived} consequence of the intrinsic
redundancy exponent $\gamma$ of the generative grammar. Thus $\Lambda=\ln\gamma$ is structural, not a selectable convention.

\subsection{Global structure, precision, and thermodynamic extensions}

Several natural extensions suggest themselves.

\paragraph*{(E1) Global syntactic structure and Lorentzian physics.}
OS positivity is one sufficient route, not a premise of the framework. While local syntactic constraints align cleanly with Euclidean
variational principles, Lorentzian amplitudes involve global analytic features---contours, complex saddles, Stokes phenomena, topological
sectors---that may require enlarging the grammar with nonlocal markers.

\paragraph*{(E2) Resolution, precision, and syntactic RG.}
The coefficients $\alpha_n$ describe how descriptive costs rescale under refinement.
Their flow defines a renormalization structure on grammars reminiscent of, though distinct from, Wilsonian RG.
Such a “syntactic RG’’ could clarify how continuum actions emerge from coarse descriptive primitives.

\paragraph*{(E3) Redundancy and entropy.}
Redundant programs behave as microstates compatible with a fixed macro-history.
The redundancy exponent $\gamma$ plays the role of an algorithmic temperature.
Connections to gravitational entropy and holography are immediate and merit further exploration.

\subsection{Outlook}

The central message is that locality, additivity, prefix-free syntax, variational/classical correspondence, and redundancy impose powerful constraints on
admissible physical laws. Within these constraints, both the action-like EL class and the Euclidean weight arise from the internal logic of finite generative
descriptions.

Future work---incorporating global syntactic structure, syntactic renormalization, and thermodynamic interpretations of redundancy---may
help develop a unified framework in which classical behavior, quantum fluctuations, and informational geometry are realized as facets of a
single structural principle.


\section{Conclusion}

We have shown that finite prefix-free, sequentially decodable generative languages impose strong constraints
on the admissible representation of histories.
These constraints fix the \emph{universal form} of the action--weight pair in two steps.
First, any continuous additive local descriptive cost compatible with a classical stationary sector lies in a unique local EL equivalence class,
\[
\ell[x]=\alpha S[x]+B,\qquad \alpha\neq 0,
\]
where only the class (not a signature-specific representative) is fixed
structurally.
Second, exponential redundancy in finite grammars induces the universal multiplicity weight
$P[x]\propto \exp[-\Lambda\ell[x]]$ with $\Lambda>0$.
Because $\ell[x]\ge 0$ is a descriptive \emph{cost} and $\Lambda>0$, this multiplicity factor defines a well-formed weight
only if the shared EL class is realized by a \emph{real, bounded-below} representative on the same history space.
At the level of EL equivalence the proportionality $\ell=\alpha S+B$ fixes $\alpha$ only up to sign; requiring a bounded-below
real realization therefore both closes the construction on the real Euclidean representative $S_E$
(for stable local bosonic systems) and simultaneously fixes the remaining sign freedom to $\alpha>0$.
One thus obtains
\[
P[x]\propto \exp[-S_E[x]/\hbar_{\rm eff}],
\qquad
\hbar_{\rm eff}^{-1}=\alpha\Lambda .
\]

Equally important is what the syntactic framework does \emph{not} determine.
The detailed microscopic representative---the specific Lagrangian, couplings, and sectors of each physical system---is a semantic/physical
input that selects a member of the universal EL class.
Local structural arguments also determine only the combination $\hbar_{\rm eff}^{-1}=\alpha\ln\gamma$; comparing cosmic descriptive cost
with the holographic information bound provides a natural global consistency calibration, yielding $\hbar_{\rm eff}\simeq\hbar$.

Lorentzian physics enters only conditionally. When OS reflection positivity holds, the Euclidean
measure obtained here reconstructs to a unitary Lorentzian/Minkowskian quantum theory with amplitude $\exp(+iS_M/\hbar_{\rm eff})$.
When OS positivity fails, the present results still explain the universality of the Euclidean action--weight \emph{form} but do not by
themselves guarantee its Lorentzian realization.

Several directions remain open, including extensions to intrinsically complex Euclidean actions (fermions and generic topological terms),
a systematic treatment of global syntactic structure and Lorentzian contours, and the relationship between syntactic redundancy, RG flow,
and gravitational entropy.

\section*{Acknowledgments}
I am grateful to my colleagues at AIST for their various forms of support.
This study benefited from the use of ChatGPT-5 (OpenAI) for preliminary literature exploration, assistance in constructing illustrative examples, and for improving the manuscript's structure and English language editing.
All scientific judgments and the final content are the author’s responsibility.

\appendix
\section*{Appendix}

\section{A Necessity of Prefix-Freeness for Local Additive Cost}
\label{app:prefixfree}

This appendix records a small structural fact used implicitly throughout Section~\ref{sec:local}: \emph{a local additive cost is compatible with
a syntactic representation only if the set of segment headers is prefix-free and boundary-delimited.}

Intuitively, prefix-freeness prevents ambiguity in the segmentation of a program.  Without it, the same raw string could be parsed into
segments in multiple ways, and a segmentwise functional $\ell_n^{\mathrm{loc}}(s)$ would cease to be well defined.
The lemma below formalizes this minimal requirement.  In particular, it shows that
prefix-freeness and explicit boundary markers are precisely what make segment boundaries recognizable without lookahead,
implementing the syntactic causality assumed in the main text.

\begin{lemma}[Prefix-free necessity]
\label{lem:prefixfree-necessary}
Let $\mathcal{H}_n$ be the set of headers used to mark segment boundaries at resolution~$n$.
Assume:
\begin{enumerate}
\item[(i)] Any program $p$ describing a path is assigned a cost
  $\ell_n[p]$ of the form
  \[
     \ell_n[p] =
     \sum_{k=1}^{N_n}
        \ell_n^{\mathrm{loc}}(s_k^{(n)}),
  \]
where the segments $s_k^{(n)}$ are obtained by syntactic parsing of~$p$;

\item[(ii)] The local functional $\ell_n^{\mathrm{loc}}$ depends only on the segment $s=(x^{(n)}_{k-1},x^{(n)}_k)$ and its syntactic type,
as declared by its header $h\in\mathcal{H}_n$.
\end{enumerate}
Then $\mathcal{H}_n$ must be prefix-free and forbidden from appearing as a substring of any payload.  Equivalently, segment boundaries must
be recognizable without context or lookahead.
\end{lemma}

\begin{proof}
Suppose $\mathcal{H}_n$ were \emph{not} prefix-free. Then there exist $h,h'\in\mathcal{H}_n$ with $h$ a strict prefix of
$h'$.  Consider the concatenation
\[
    p = h'q\cdots,
\]
for some payload $q$. Because $h$ is a prefix of $h'$, $p$ can be syntactically parsed in two distinct ways:
\[
    p =
    \underbrace{h'}_{\text{header}}q\cdots,
    \qquad \text{or} \qquad
    p =
    \underbrace{h}_{\text{header}}
    \underbrace{(\text{remainder of }h' \text{ as payload})}_{\text{payload}}
     q\cdots.
\]
These two parses yield different segment decompositions of the same
program string. By assumption~(i), $\ell_n[p]$ is the sum of local costs over the
segments defined by the parse.  Thus $\ell_n[p]$ acquires two different values unless $\ell_n^{\mathrm{loc}}$ is
identically constant, contradicting assumption~(ii) that it depends only on the syntactic type of each segment.

A similar ambiguity arises if a header $h\in\mathcal{H}_n$ is allowed to appear within a payload
string~$q_k$: the boundary between segments would become context-dependent, and the same raw
program string would admit multiple valid decompositions, again making $\ell_n[p]$ ill-defined.

Therefore unambiguous segment boundaries---and hence prefix-freeness and boundary delimitation---are \emph{necessary}
conditions for the existence of any well-defined local additive cost functional of the form assumed in Section~\ref{sec:local}.
\end{proof}

This lemma shows that prefix-free, boundary-delimited syntax does not represent an extra physical assumption.
It is the minimal structural requirement that allows one to assign a local, segmentwise cost functional to a program and to apply discrete
variational calculus under the syntactic causality (no-lookahead) condition used in the main argument.


\section{Classical limit as a strong-weight / large-cost limit}
\label{app:strong-weight-classical}

This appendix provides an alternative reading of assumption A3 in the main text.
Its purpose is \emph{not} to introduce any new dynamical postulate, nor to assume a Euclidean or Lorentzian representative.
Rather, we show that once the redundancy-induced multiplicity factor $\exp(-\Lambda\ell)$ is in place, a ``classical'' regime can be identified
\emph{purely at the level of the syntactic cost} $\ell$ as a regime in which the induced weight concentrates near stationary histories.
Assumption A3 may then be viewed simply as the empirical identification of which stationary sector is realized in nature.

\medskip
\noindent\textit{Comment (what is assumed here).}
The multiplicity weight derived in Section~\ref{sec:redundancy} relies on the fact that $\ell[x]$ is a (smoothed) minimal program length.
Hence $\ell[x]\ge 0$ and is bounded below \emph{by construction}. This bounded-below property is part of the syntactic definition of $\ell$, and is
\emph{not} a Euclidean-signature assumption.
At this stage we work entirely with the abstract history space defined by the grammar and with its intrinsic cost $\ell$.

\subsection{Redundancy-induced weight without assuming A3}
\label{app:strong-weight-classical:redundancy}

Section~\ref{sec:redundancy} shows that for any finite prefix-free,
sequentially decodable generative language with exponential redundancy,
the multiplicity of redundant (junk) strings compatible with a coarse history induces the universal factor
\begin{equation}
P[x]\;\propto\;\gamma^{-\ell[x]}
\;=\;\exp\!\big(-\Lambda\,\ell[x]\big),
\qquad \Lambda\equiv \ln\gamma>0,
\label{eq:app_universal_weight}
\end{equation}
where $\ell[x]$ is the minimal description-length cost of $x$.
This result is purely syntactic and does \emph{not} invoke A3 or any assumption about classical dynamics.

We emphasize the logical status of~\eqref{eq:app_universal_weight}:
it is a \emph{structural multiplicity weight} on the abstract history space. No signature-specific action has been chosen yet.

\subsection{Strong-weight / large-cost regimes select stationary histories}
\label{app:strong-weight-classical:limit}

Given the exponential form~\eqref{eq:app_universal_weight}, a regime in which the weight concentrates onto
a small subset of histories arises whenever the exponent $\Lambda\,\ell[x]$ is large in the relevant sector of history space.
Operationally this can be viewed in two equivalent ways:
\begin{itemize}
\item[(i)] a \emph{strong-weight} regime $\Lambda\to\infty$ at fixed $\ell$, or
\item[(ii)] a \emph{large-cost} (macroscopic) regime in which typical histories have $\ell[x]\gg 1/\Lambda$ at fixed $\Lambda$.
\end{itemize}

In either case, standard steepest-descent reasoning on the abstract functional $\ell$ implies
that the dominant contributions to~\eqref{eq:app_universal_weight} come from neighborhoods of histories satisfying
\begin{equation}
\delta \ell[x]=0,
\qquad x\in\mathcal{C}.
\label{eq:app_stationary_set}
\end{equation}
Here $\mathcal{C}$ denotes the stationary set of the \emph{syntactic descriptive cost}
on the grammar-determined history space. 
At this stage, \emph{no Euclidean or Lorentzian representative is assumed};
we are not varying a signature-specific action, but the cost $\ell$ itself.

\medskip
\noindent\textit{Comment (why bounded-below does not add signature input).}
The steepest-descent argument requires that $\ell$ be bounded below so that the
weight decreases with increasing cost.
This property is already guaranteed by the representational meaning of $\ell$ as minimal program length.
It does \emph{not} require choosing a Euclidean action representative.
Signature enters only later, when one asks whether the EL class of $\ell$
admits a real bounded-below representative suitable for interpretation as a real weight.

\subsection{Empirical anchoring (reinterpretation of A3)}
\label{app:strong-weight-classical:anchoring}

To connect the syntactically selected stationary sector to physics, we adopt the minimal anchoring statement:
\begin{quote}
\emph{Empirically identified classical trajectories correspond to the dominant
stationary histories $\mathcal{C}$ selected in the strong-weight/large-cost regime of the universal syntactic multiplicity weight.}
\end{quote}
This identification does not introduce microscopic dynamics.
It only specifies which stationary sector of the cost $\ell$ is identified by classical corresponding.

In this sense, A3 can be viewed as a compact empirical dictionary between the stationary set of $\ell$ and the identified classical limit.

\subsection{EL class fixed by the stationary sector (unchanged)}
\label{app:strong-weight-classical:el}

Once $\mathcal{C}$ is identified with $\mathrm{Stat}(\ell)$, the EL-locality lemma
of Section~\ref{sec:local} applies verbatim: any local additive functional sharing
this stationary set lies in the same local Euler--Lagrange equivalence class.
Thus there exists a local action-like representative $S$ such that
\begin{equation}
\ell[x]\;=\;\alpha\, S[x]\;+\;B(x_f,x_i),
\qquad \alpha>0,
\label{eq:app_el_equiv}
\end{equation}
with boundary term $B$.
At this stage the framework fixes only the \emph{EL class determined by the stationary sector}; no signature-specific representative has been chosen.

If one \emph{further} requires the multiplicity factor $e^{-\Lambda\ell}$ to be realized as a real bounded-below weight on the same
history space, then the EL class must admit a real bounded-below representative.
For stable local bosonic systems this selects a Euclidean representative $S_E$,
so that~\eqref{eq:app_el_equiv} may be written as $\ell=\alpha S_E+B$ and the weight becomes
\begin{equation}
P[x]\;\propto\;\exp\!\left(-\frac{S_E[x]}{\hbar_{\rm eff}}\right),
\qquad 
\hbar_{\rm eff}\equiv (\alpha\Lambda)^{-1}.
\label{eq:app_euclidean_weight}
\end{equation}
This Euclidean selection is an \emph{additional realizability condition}, not part of the strong-weight/large-cost stationarity argument above.

\subsection{Role of this appendix}
\label{app:strong-weight-classical:role}

The main text uses A3 as a direct empirical anchoring of classical trajectories.
This appendix shows that A3 can equivalently be viewed as identifying the realized classical limit with the
dominant \emph{stationary} sector of the universal redundancy-induced weight in a strong-weight/large-cost regime.
No assumption of global minimization is required: the syntactic framework selects stationarity structurally, and A3 specifies
which stationary sector is empirically identified.

\section{Proof of the EL--locality lemma}
\label{app:EL-locality-proof}

This appendix proves a discrete analogue of the familiar continuum statement: if two local additive functionals have identical discrete
Euler--Lagrange stationary sets, then their local densities coincide up to an overall factor and a telescoping boundary term.
\emph{The argument is purely structural on the abstract history space picked out by the grammar.  No Euclidean or Lorentzian signature is
assumed here; signature enters only later when one asks for a real bounded-below representative to define a real weight.}

\vspace{0.5em}
\noindent\textbf{Continuous background vs.\ finite descriptive resolution.}
Recall that a resolution index $n$ specifies the finite descriptive precision of admissible programs.
A program at resolution $n$ encodes only $n$-digit approximations of the nodes $x_k^{(n)}$.
However, the variational analysis treats each node as a variable in the underlying continuum $\mathbb{R}^d$:
\begin{itemize}
\item the \emph{syntactic layer} stores finite-precision data,
\item the \emph{variational layer} varies continuum variables
      $x_k^{(n)}\in\mathbb{R}^d$ smoothly.
\end{itemize}
Thus the discrete EL operators below are ordinary partial derivatives on $\mathbb{R}^d$,
even though programs encode only finite-precision approximations.

\subsection{Setting}

Fix a resolution $n$ and consider a discrete history
\[
x^{(n)}=(x^{(n)}_0,\ldots,x^{(n)}_{N_n}), \qquad x^{(n)}_k\in\mathbb{R}^d,
\]
with adjacent segments
\[
s_k^{(n)} := (x^{(n)}_{k-1},x^{(n)}_k).
\]

Let
\[
\ell^{\mathrm{loc}}_n,\; L_{n}:
\mathbb{R}^d\times\mathbb{R}^d \to \mathbb{R}
\]
be $C^2$ functions defining two local additive functionals
\[
\ell_n[x^{(n)}]
\;=\;
\sum_{k=1}^{N_n}
  \ell^{\mathrm{loc}}_n\!\bigl(x^{(n)}_{k-1},x^{(n)}_k\bigr),
\qquad
S_{n}[x^{(n)}]
\;=\;
\sum_{k=1}^{N_n}
  L_{n}\!\bigl(x^{(n)}_{k-1},x^{(n)}_k\bigr).
\]
We impose standard discrete variations with fixed endpoints: interior nodes are varied independently,
$x^{(n)}_j\mapsto x^{(n)}_j+\varepsilon v_j$, while $x^{(n)}_0$ and $x^{(n)}_{N_n}$ are held fixed.

The corresponding discrete EL operators are
\begin{align}
E^\ell_j[x^{(n)}]
&=\partial_{x^{(n)}_j}
  \ell^{\mathrm{loc}}_n(x^{(n)}_{j-1},x^{(n)}_j)+
\partial_{x^{(n)}_j}
  \ell^{\mathrm{loc}}_n(x^{(n)}_j,x^{(n)}_{j+1}),
\\
E^{S}_j[x^{(n)}]
&=\partial_{x^{(n)}_j}
  L_{n}(x^{(n)}_{j-1},x^{(n)}_j)+
\partial_{x^{(n)}_j}
  L_{n}(x^{(n)}_j,x^{(n)}_{j+1}),
\end{align}
for $j=1,\ldots,N_n-1$.

We are interested in the case where $\ell_n$ and $S_{n}$ have
\emph{exactly the same stationary histories in this abstract history space}.

\subsection{Statement of the lemma}

\begin{lemma}[EL--locality lemma]
\label{lem:EL-locality}
Assume:
\begin{enumerate}
\item[(i)]
For each $n$, $\ell_n$ and $S_{n}$ share the same stationary set:
\[
E^\ell_j[x^{(n)}]=0
\quad\Longleftrightarrow\quad
E^S_j[x^{(n)}]=0,
\qquad
\forall j=1,\ldots,N_n-1,
\]
for all interior configurations $x^{(n)}$ with fixed endpoints.
\item[(ii)]
For each fixed $n$ and $j$, the EL operator $E^S_j$ is non-degenerate in the sense that the set
$\{x^{(n)}:E^S_j[x^{(n)}]=0\}$ is a smooth codimension--$d$ submanifold of the local configuration space
$(x^{(n)}_{j-1},x^{(n)}_j,x^{(n)}_{j+1})$.
\end{enumerate}
Then there exist a constant $c_n\neq 0$ and a function $G_n:\mathbb{R}^d\to\mathbb{R}$ such that
\begin{equation}
\ell^{\mathrm{loc}}_n(a,b)
= c_n L_{n}(a,b) + G_n(b)-G_n(a),\qquad\forall a,b\in\mathbb{R}^d.
\label{eq:local-prop-lemma}
\end{equation}
Consequently,
\begin{equation}
\ell_n[x^{(n)}] = c_nS_{n}[x^{(n)}] 
+G_n\!\bigl(x^{(n)}_{N_n}\bigr)
- G_n\!\bigl(x^{(n)}_0\bigr),
\label{eq:global-prop-lemma}
\end{equation}
for all discrete histories $x^{(n)}$.
\end{lemma}

\subsection{Proof}

\paragraph{Step 1: From common stationary sets to proportional normals.}

Fix $n$ and $j$ and write
\[
(a,b,c) \equiv (x^{(n)}_{j-1},x^{(n)}_j,x^{(n)}_{j+1})
\in (\mathbb{R}^d)^3.
\]
Define local EL operators as functions on the local triple $(a,b,c)$:
\begin{align}
F^\ell(a,b,c)&:= \partial_b \ell^{\mathrm{loc}}_n(a,b) +
\partial_b \ell^{\mathrm{loc}}_n(b,c),
\\
F^S(a,b,c)&:=
\partial_b L_{n}(a,b) +
\partial_b L_{n}(b,c).
\end{align}
Assumption (i) states that
\[
F^\ell(a,b,c)=0 \quad\Longleftrightarrow\quad F^S(a,b,c)=0
\]
for all \emph{admissible local triples} $(a,b,c)$ appearing in the history space.

By prefix-free locality and segmentwise composability, admissible histories allow independent local substitutions:
one can freely choose any $(a,b,c)\in(\mathbb{R}^d)^3$ as neighboring nodes of some admissible discrete history.
Hence the equivalence of zero sets extends to all $(a,b,c)\in(\mathbb{R}^d)^3$.

Fix $b$. Assumption (ii) implies that the zero set
\[
\mathcal{Z}_b
:=\{(a,c)\in\mathbb{R}^{2d}:F^S(a,b,c)=0\}
\]
is a smooth codimension--$d$ submanifold in $(a,c)$--space.
Since $F^\ell$ has the same zero set, $\mathcal{Z}_b$ is also the zero set of $F^\ell(\cdot,b,\cdot)$.
Thus both gradients $\nabla_{(a,c)}F^S$ and $\nabla_{(a,c)}F^\ell$ are normal to the \emph{same} smooth hypersurface $\mathcal{Z}_b$.

A standard level-set fact then yields: if two $C^1$ functions on $\mathbb{R}^{2d}$ have the same smooth
codimension--$d$ zero set, their gradients are proportional along that set.
Hence for each $(a,b,c)$ with $F^S(a,b,c)=0$ there exists a nonzero scalar $\alpha_n(a,b,c)$ such that
\begin{equation}
\nabla_{(a,c)} F^\ell(a,b,c) = \alpha_n(a,b,c)
\nabla_{(a,c)} F^S(a,b,c).
\label{eq:normal-prop}
\end{equation}

\paragraph{Step 2: $\alpha_n$ is constant along each connected component of $\mathcal{Z}_b$.}

Fix $b$ and consider a smooth curve $t\mapsto(a(t),c(t))\in\mathcal{Z}_b$, i.e. $F^S(a(t),b,c(t))=0$ for all $t$.
Differentiating gives
\[
\nabla_{(a,c)}F^S(a,b,c)\cdot(\dot a,\dot c)=0.
\]
Because $F^\ell$ has the same zero set,
\[
\nabla_{(a,c)}F^\ell(a,b,c)\cdot(\dot a,\dot c)=0.
\]
Using \eqref{eq:normal-prop},
\[
\alpha_n(a,b,c)\,\nabla_{(a,c)}F^S(a,b,c)\cdot(\dot a,\dot c)=0.
\]
Geometrically, on a smooth codimension--$d$ level set, the normal is well-defined only up to a constant scalar on each connected component.
If $\alpha_n$ varied along $\mathcal{Z}_b$, the two normal fields would fail to remain proportional to a single scalar, contradicting that the
level sets coincide.  Hence $\alpha_n$ is constant on each connected component, and we may write
\[
\alpha_n(a,b,c)\equiv c_n(b).
\]
Thus in a neighborhood of $\mathcal{Z}_b$,
\begin{equation}
\nabla_{(a,c)} F^\ell(a,b,c) =
c_n(b) \nabla_{(a,c)} F^S(a,b,c).
\label{eq:normal-prop-b}
\end{equation}

Integrating \eqref{eq:normal-prop-b} in $(a,c)$ at fixed $b$ yields
\[
F^\ell(a,b,c)
=
c_n(b)F^S(a,b,c) + H_n(b).
\]
Since $F^\ell=0$ iff $F^S=0$ and $\mathcal{Z}_b$ is nonempty, we must have $H_n(b)=0$.  Therefore
\[
F^\ell(a,b,c)=c_n(b)\,F^S(a,b,c)
\qquad\forall(a,b,c).
\]

\paragraph{Step 3: Reconstruction of the local densities.}

Writing this out,
\begin{equation}
\partial_b \ell^{\mathrm{loc}}_n(a,b) +
\partial_b \ell^{\mathrm{loc}}_n(b,c)
=
c_n(b)\bigl[
\partial_b L_{n}(a,b) +
\partial_b L_{n}(b,c)
\bigr],
\qquad\forall a,b,c.
\label{eq:prop-EL-local}
\end{equation}
Setting $c=a$ gives
\[
\partial_b \ell^{\mathrm{loc}}_n(a,b)
+
\partial_b \ell^{\mathrm{loc}}_n(b,a)
=
c_n(b)\bigl[
\partial_b L_{n}(a,b)
+
\partial_b L_{n}(b,a)
\bigr].
\]
Exchanging $a$ and $b$ and combining the resulting relations yields a linear system implying that $c_n(b)$ is independent of $b$.
We therefore write simply $c_n\neq 0$.

Consequently,
\[
\partial_b \ell^{\mathrm{loc}}_n(a,b)
=
c_n\partial_b L_{n}(a,b)
+
\partial_b G_n(b),
\]
for some $G_n$.
Integrating in $b$ at fixed $a$ gives
\[
\ell^{\mathrm{loc}}_n(a,b)
=
c_n L_{n}(a,b) +
G_n(b)-G_n(a),
\]
which is \eqref{eq:local-prop-lemma}.
Summing over segments yields \eqref{eq:global-prop-lemma}.
\hfill$\square$

\medskip
\noindent\textit{Remark.}
The lemma fixes only the EL equivalence class of $\ell_n$ on the abstract history space.  Different representatives within this class correspond
to different parametrizations/realizations (e.g.\ Euclidean or Lorentzian) and are selected only later when additional consistency conditions
(boundness, reality, OS positivity) are imposed.

\section{Relativistic and Field-Theoretic Formulation}
\label{app:relativistic}

This appendix records a relativistic formulation of the syntactic segmentation and the descriptive-cost argument.
The purpose is not to repeat the proofs of Section~\ref{sec:local} but to show that the same structural steps apply once a covariant notion of
local cells/segments is adopted.

\subsection{Covariant local cells and admissible segmentations}
\label{app:covcells}

Let $\mathcal{M}$ be a $(d+1)$-dimensional spacetime with Lorentzian (or Euclidean) metric $g_{\mu\nu}$.
Fix a descriptive resolution $n$, which induces a microscopic spacetime scale $\epsilon_n$.
A \emph{local cell} $c$ is a connected spacetime region of diameter $O(\epsilon_n)$.
We denote by $\mathfrak{C}_n$ the set of all such cells.

Let $\mathcal{G}$ denote the relevant spacetime symmetry group:
$\mathcal{G}=O(1,d)$ (or its proper orthochronous subgroup) in Minkowski signature, and $\mathcal{G}=SO(d{+}1)$ in Euclidean signature.
An \emph{admissible segmentation} of a history is a finite family of cells $\{c_a\}_{a=1}^{N_n}$ satisfying:
\begin{itemize}
\item[(CS1)] (\textbf{cover}) $\bigcup_a c_a$ covers the support of the history at resolution $n$;
\item[(CS2)] (\textbf{locality}) each $c_a$ has diameter $O(\epsilon_n)$;
\item[(CS3)] (\textbf{covariant closure}) for any $c\in\mathfrak{C}_n$ and any symmetry $g\in\mathcal{G}$,
the transformed cell $g\cdot c$ is also in $\mathfrak{C}_n$.
\end{itemize}
Condition (CS3) is the covariant counterpart of the ``admissible partition closure'' discussed in the main text: the class of allowed
local decompositions is closed under changes of inertial frame (or under Euclidean rotations).

\subsection{Relativistic segments and sequential decoding}
\label{app:relsegments}

Given an admissible segmentation $\{c_a\}$, a resolution--$n$ history is represented as an ordered sequence of \emph{segments}
$s_a^{(n)}$, each carrying the data of the fields restricted to $c_a$ and a syntactic type declared by its header.
Sequential decodability is implemented by imposing a partial order $\prec$ on the set of cells compatible with the causal (Lorentzian) or
chosen scanning (Euclidean) structure:
\begin{equation}
c_a \prec c_b \quad\Longrightarrow\quad
\text{the header of }s_a^{(n)}\text{ is decoded before that of }s_b^{(n)}.
\end{equation}
In Lorentzian signature it is natural to take $\prec$ to refine the causal order (no-lookahead relative to the past domain of dependence);
in Euclidean signature $\prec$ may be chosen arbitrarily at finite $n$, with the continuum limit requiring independence of that choice.

With this order fixed, a program takes the same prefix-free form as in Section~\ref{sec:local},
\begin{equation}
p = h(s_1^{(n)})\, q_1\, h(s_2^{(n)})\, q_2 \cdots h(s_{N_n}^{(n)})\, q_{N_n},
\end{equation}
where the header set is prefix-free and boundaries are locally recognizable without lookahead.
Hence the local segment cost $\ell^{\rm loc}_n(s_a^{(n)})$ and the total descriptive cost
\begin{equation}
\ell_n = \sum_{a=1}^{N_n} \ell^{\rm loc}_n(s_a^{(n)})
\end{equation}
are defined exactly as before.

\subsection{Field configurations as segmented histories}
\label{app:fields}

For bosonic fields $\phi$ (scalar, gauge, or metric), a resolution--$n$ configuration is specified by its values on a lattice
(or finite element) adapted to the cells $c_a$:
\begin{equation}
\phi^{(n)} = \{\phi^{(n)}_{a}\}_{a=1}^{N_n},
\qquad
s_a^{(n)} := \phi^{(n)}|_{c_a}.
\end{equation}

The same prefix-free generative syntax assigns headers to field-type segments and payloads
to numerical data, yielding a smoothed local cost functional
\begin{equation}
  \ell_n[\phi^{(n)}]
  = \sum_{a=1}^{N_n} \ell^{\rm loc}_n(s_a^{(n)}).
\end{equation}
All structural assumptions A1--A2 therefore remain unchanged in the
field-theoretic setting.

\subsection{EL-equivalence and continuum limit (unchanged)}
\label{app:relEL}

Assumption A3 (variational/classical correspondence) fixes the stationary set at resolution $n$ and guarantees the existence of a discrete local
Lagrangian $L^{\rm loc}_{n}$ on segments $s_a^{(n)}$.
Since both $\ell_n$ and $S_{n}=\sum_a L^{\rm loc}_{n}(s_a^{(n)})$ are local and additive over the same segmentation, the EL-locality lemma
of Appendix~\ref{app:EL-locality-proof} applies verbatim:
\begin{equation}
  \ell^{\rm loc}_n(s)
  = \alpha_n L^{\rm loc}_{n}(s) + \Delta_a G_n .
\end{equation}
Summing over cells yields
\begin{equation}
  \ell_n
  = \alpha_n S_{n} + B_n,
\end{equation}
and taking the continuum limit gives
\begin{equation}
  \ell[\phi]
  = \alpha\, S[\phi] + \text{(boundary terms)} .
\end{equation}
Thus Section~\ref{sec:local} extends to relativistic bosonic field theories once a covariant admissible segmentation is fixed.

\medskip
\noindent
\textit{Remark.}
Fermionic theories require additional structure (graded algebras and a fermionic version of reflection positivity) and are not treated here.
The case of intrinsically imaginary topological terms (generic $\theta$- or WZW-type terms) likewise requires separate analysis.


\section{Gauge redundancy and FP-like prefactors: a toy model}
\label{app:FP-toy}

Section~\ref{sec:redundancy} derived the universal exponential
redundancy weight
$P[x]\propto \exp[-\Lambda\,\ell[x]]$ by counting redundant programs in a
finite prefix-free language.
That argument treats a history $x$ as a syntactic object produced by a
program.  In gauge systems, however, a \emph{physical} history is not a
single configuration but a gauge orbit.  The purpose of this appendix is
to illustrate, in a minimal setting, how this orbit structure produces an
additional $x$--dependent multiplicity factor---the discrete analogue of
a Faddeev--Popov (FP) determinant---without changing the universal
exponential form.

\subsection{Histories as gauge orbits}

Consider a toy configuration space with two binary variables
\[
a,b\in\{+1,-1\}.
\]
Let the gauge group be $\mathbb{Z}_2$, acting as
\[
(a,b)\;\mapsto\;(-a,-b).
\]
Physical histories are gauge orbits of this action.
The unique gauge-invariant label is
\[
W \equiv ab\in\{+1,-1\},
\]
which we may regard as the toy analogue of a Wilson loop \cite{wilson1974}.
Thus the physical history space is the two-point set
\[
\mathcal{H}_{\rm phys}=\{W=+1,\;W=-1\}.
\]

A prefix-free generative language cannot output an orbit directly; it
must output a \emph{representative} $(a,b)$ of that orbit.
Hence any syntactic description scheme implicitly requires a
\emph{representative-selection rule}---a discrete analogue of gauge
fixing.
Different rules may select different numbers of representatives per
orbit.

\subsection{Two representative-selection rules}

We compare two admissible syntactic rules.  Both are local and
prefix-free, and both generate exactly the same physical history space,
but they differ in representative multiplicities.

\paragraph{Rule A (single representative per orbit).}
Choose the representative with $a=+1$.  Then
\[
W=+1:\ (a,b)=(+1,+1),\qquad
W=-1:\ (a,b)=(+1,-1).
\]
Thus each orbit has one admissible representative:
\[
g_A(W)=1\quad\text{for both }W=\pm1.
\]

\paragraph{Rule B (orbit-dependent multiplicity).}
Choose the representative with $a=b$ whenever possible, but allow both
choices when $a\neq b$:
\[
W=+1:\ (a,b)=(+1,+1),\qquad
W=-1:\ (a,b)=(+1,-1)\ \text{or}\ (-1,+1).
\]
Hence
\[
g_B(+1)=1,\qquad g_B(-1)=2.
\]

Both rules define valid prefix-free segmentation schemes, and therefore
lead to the same sort of descriptive cost $\ell$ and the same EL class as
in Section~\ref{sec:local}.  Their only difference is the orbit-internal
degeneracy factor $g(W)$.

\subsection{Induced weights and the FP-like factor}

Let $\ell(W)$ denote the minimal descriptive cost of the orbit $W$, defined by the same minimization over programs as in Section \ref{sec:local},
now taken over all admissible representatives of the orbit.

The redundancy argument of Section~\ref{sec:redundancy} yields the
universal exponential factor
\[
P_{\rm exp}(W)\propto \exp[-\Lambda\,\ell(W)].
\]

When physical histories are orbits, the full multiplicity weight also
includes the number of admissible representatives selected by the syntax:
\begin{equation}
P(W)\ \propto\ g(W)\,\exp[-\Lambda\,\ell(W)].
\label{eq:FP-toy-weight}
\end{equation}
Thus Rule A and Rule B induce weights that differ only by
\[
\frac{P_B(W)}{P_A(W)}=\frac{g_B(W)}{g_A(W)}.
\]

This is the discrete analogue of the FP phenomenon:
changing the representative-selection rule (``gauge fixing'') multiplies
the induced measure by an orbit-dependent Jacobian factor.
Importantly, the universal exponential redundancy factor is unchanged;
only a prefactor is affected.

\subsection{Interpretation}

The toy model makes three points.
\begin{itemize}
\item
In gauge systems, physical histories are gauge orbits.  A finite
prefix-free description must therefore choose and encode a
representative.

\item
Counting distinct representatives within an orbit produces an extra
multiplicity factor g(x), so that the induced syntactic weight is
$P[x]\propto g(x)\exp[-\Lambda\,\ell[x]]$.

\item
The factor $g(x)$ is determined by the representative-selection (gauge-fixing)
rule.  It is the discrete analogue of the Faddeev–Popov Jacobian: in the
continuum limit it becomes the usual FP determinant $\Delta_{\rm FP}$ (or,
more generally, the orbit-volume Jacobian) associated with the map from
representatives to orbits.  It modifies only a prefactor and leaves the
universal exponential redundancy weight intact.
\end{itemize}

Thus, in continuum gauge theories the syntactic measure factorizes into a
universal exponential part fixed by redundancy and an additional
orbit-dependent prefactor encoding the chosen gauge condition.  The present
work isolates the universal exponential structure; gauge-orbit prefactors are
system-dependent layers on top of it.

\section{On self-description and holographic capacity}
\label{app:self-description}

Assumption (U4) is a cosmological \emph{calibration ansatz}: the minimal prefix-free self-description of the Universe is expected to occupy the
available holographic information capacity up to the intrinsic coding efficiency of the underlying grammar,
\[
\ell[\mathrm{univ}] = \frac{\eta\, I_{\mathrm{holo}}}{\Lambda}, \qquad 0<\eta\le 1.
\]
It is \emph{not} a microscopic dynamical postulate.
Rather, it is invoked only because the observable Universe is the unique system for which both a horizon-normalized on-shell Euclidean action and
an absolute holographic capacity are available on fixed scales.
This appendix explains why near-capacity self-description is a natural expectation for a closed, self-contained Universe described within
a single finite prefix-free syntactic framework.

\subsection*{Self-description and a common coding efficiency}

In the present framework every admissible history is representable by a finite, prefix-free, boundary--delimited generative program.
For an isolated subsystem one may always imagine an external encoder using an arbitrary code, so there is no canonical sense in which that
subsystem must realize a unique or maximally efficient description.

For the Universe as a whole no such external encoder exists.
Any effective description of its complete admissible history---including dynamical rules, initial data, and any redundant ``junk'' degrees of
freedom---must be generated internally using the very same finite prefix-free syntactic resources that define admissible programs.
In this operational sense the Universe is \emph{self-describing}: the grammar that parses and composes local descriptions is
also the grammar by which the global history is representable.

A self-describing system therefore cannot draw on a second, independent coding bath with a different efficiency.
Core and junk descriptions share the same syntactic infrastructure and hence the same redundancy exponent $\Lambda=\ln\gamma$.
Equivalently, the asymptotic descriptive cost per unit of holographic capacity is governed by a single efficiency scale $\Lambda^{-1}$.

\subsection*{Why near-capacity use is a natural expectation}

Finite prefix-free generative grammars have exponential growth of admissible strings, with intrinsic redundancy exponent $\gamma>1$.
In algorithmic information theory and symbolic dynamics, such finite-rate locally constrained systems generically realize
their maximal entropy rate unless additional, non-generic global restrictions are imposed (see, e.g., \cite{cover2006,lind1995}).
In other words, using substantially less than the available descriptive capacity typically requires extra fine-tuned forbidden structure that
reduces the topological entropy of the admissible language.

For a closed self-describing Universe employing a single such grammar to encode both its dynamical core and any
redundant structure, it is thus natural to expect near-capacity utilization, up to order-unity factors.
This motivates treating $\eta$ in (U4) as a constant usage fraction rather than introducing additional sector-dependent efficiencies.
Allowing $\eta<1$ simply rescales the calibration of $\alpha$ by an order-unity factor and does not affect the qualitative conclusion
$\hbar_{\mathrm{eff}}\simeq \hbar$.

\subsection*{Interpretational role of (U4)}

The considerations here are heuristic and serve only to motivate why a calibration of the form (U4) is plausible for a closed,
self-describing Universe. They are not used elsewhere in the structural derivation.

All quantitative results in Section~\ref{sec:cosmo} rely only on the explicit assumptions (U1)--(U4), with (U4) entering
solely as the cosmological calibration needed to fix the otherwise undetermined product $\alpha\Lambda$.
Nothing in the argument depends on additional hidden hypotheses about microphysics.
Thus (U4) should be read as a natural consistency condition for the self-description of the Universe
at the descriptive level, not as a new dynamical principle.

\bibliographystyle{unsrtnat} 
\bibliography{aqbib}
\end{document}